\documentclass[a4paper,12pt]{article}

\usepackage{test} 
\usepackage{color}
\usepackage[colorlinks,citecolor=blue,urlcolor=magenta]{hyperref}
\usepackage{doi}
\usepackage[authoryear,round,longnamesfirst]{natbib}

\usepackage[inline,shortlabels]{enumitem}
\setlist[enumerate,1]{label=(\roman*)}
\usepackage{booktabs}
\usepackage{longtable}
\usepackage{subcaption}

\usepackage{graphicx}

\usepackage{filecontents} 
\usepackage[margin = 1.25in]{geometry}
\usepackage{setspace}
\title{Capital and Labor Income Pareto Exponents across Time and Space\thanks{We thank Xavier Gabaix, Yannick Hoga, and Makoto Nirei for comments and suggestions.}}
\author{Tjeerd de Vries\thanks{Department of Economics, University of California San Diego. Email: \href{mailto:tjdevrie@ucsd.edu}{tjdevrie@ucsd.edu}.} \and Alexis Akira Toda\thanks{Department of Economics, University of California San Diego. Email: \href{mailto:atoda@ucsd.edu}{atoda@ucsd.edu}.}}

\numberwithin{equation}{section}

\newcommand{\capital}{\text{cap}}
\newcommand{\labor}{\text{lab}}
\renewcommand{\hat}{\widehat} 
\renewcommand{\Pr}{\operatorname{P}}

\defcitealias{LIS2021}{LIS}

\newcommand{\Nmin}{1,000\xspace} 
\newcommand{\ncountries}{52\xspace} 
\newcommand{\Ncross}{475\xspace} 
\newcommand{\medcap}{1.46\xspace} 
\newcommand{\medlab}{3.35\xspace} 
\newcommand{\Ncap}{342\xspace}
\newcommand{\Nreject}{294\xspace}
\newcommand{\rejectPercent}{86\%\xspace}

\begin{document}
\maketitle


\begin{abstract}
We estimate capital and labor income Pareto exponents across \Ncross country-year observations that span \ncountries countries over half a century (1967--2018). We document two stylized facts: (i) capital income is more unequally distributed than labor income in the tail; namely, the capital exponent (1--3, median \medcap) is smaller than labor (2--5, median \medlab), and (ii) capital and labor exponents are nearly uncorrelated. To explain these findings, we build an incomplete market model with job ladders and capital income risk that gives rise to a capital income Pareto exponent smaller than but nearly unrelated to the labor exponent. Our results suggest the importance of distinguishing income and wealth inequality.

\medskip

{\bf Keywords:} income fluctuation problem, inequality, power law.

\medskip

{\bf JEL codes:} C46, D15, D31, D52.
\end{abstract}


\section{Introduction}

The purpose of this paper is to estimate and document the Pareto exponents for capital and labor income separately for as many countries and years as possible. We say that a positive random variable $X$ obeys a power law with Pareto exponent $\alpha>0$ if the tail probability decays like a power function: $\Pr(X>x)\sim x^{-\alpha}$ for large $x$.\footnote{\label{fn:RV}More precisely, we say that a random variable $X$ has a Pareto upper tail with exponent $\alpha>0$ if $\Pr(X>x)=x^{-\alpha}\ell(x)$ for some slowly varying function $\ell$. A function $\ell:(0,\infty)\to \R$ is said to be \emph{slowly varying (at infinity)} if it is nonzero for sufficiently large $x$ and $\lim_{x\to\infty}\ell(tx)/\ell(x)=1$ for each $t>0$. See \cite{BinghamGoldieTeugels1987} for a comprehensive treatment of the theory of regular variation.} In the context of the income distribution, the Pareto exponent characterizes the tail heaviness of high incomes and hence top tail inequality. We remain agnostic about the shape of the income distribution away from the tail. Our study is motivated by the following two observations. First, we are not aware of a comprehensive study that documents the capital and labor income Pareto exponents separately for many countries and years, despite their importance. Second, the Pareto exponent has desirable properties relative to other popular inequality measures such as the Gini coefficient or top income shares.

Consider the first point. Conceptually, capital and labor income are different entities. While the former is the return for providing capital (wealth), the latter is the return for providing labor services, and there is no particular reason to expect a relation between the two. Although these two forms of income are conceptually distinct, it is often put together as just ``income'' and discussed in the context of inequality and related policies. If capital and labor income are quantitatively different, a policy design based on total income may be misleading. To give one example, consider the theory of optimal taxation \citep{Saez2001}, where the income Pareto exponent plays an important role. \cite{SaezStancheva2018} carefully distinguish capital and labor income and apply the theory of optimal taxation in the United States. They find that with an income elasticity of $e=0.5$, the optimal top marginal tax rate is about 50\% for labor and 60\% for capital (see their Figure 5). This difference directly comes from the fact that capital and labor income Pareto exponents are distinct. Thus, distinguishing capital and labor income inequality is potentially important for policy designs.

Consider the second point about the desirable properties of Pareto exponents. In the applied literature such as \cite{Piketty2003} and \cite{piketty-saez2003}, top income shares (such as the top 1\% income share) are more commonly reported than the income Pareto exponent, perhaps because top shares are summary statistics that can be computed without specifying functional forms or can be understood by non-experts without special knowledge of statistics. However, \cite{atkinson2005comparing} documents methodological problems regarding the cross-country comparison of top income shares, citing the differences in tax units (\eg, individual or household) and legislation (\eg, whether social security benefits are taxable). One of the reasons such issues arise is because it is not always clear how to define the population and measure small units.\footnote{Imagine how to formally distinguish cities, towns, villages, and settlements; continents, islands, and islets; and inland seas, lakes, and ponds. How to define units and how to measure small units matter for the size distribution of population, land mass, and water surface area.} Because common inequality measures such as the Gini coefficient and top income shares require the knowledge of the entire distribution, these quantities are greatly affected by the definition and measurement of small units. Using the Pareto exponents significantly alleviates these definition and measurement issues because the Pareto distribution is scale invariant (see \citealp{JessenMikosch2006} for a summary) and its exponent depends only on the tail behavior, not the entire distribution. For example, doubling the income of all households in the top 1\% of the income distribution makes the top 1\% income share (roughly) twice as large, but the Pareto exponent is unaffected. A similar comment applies to any inequality measure that depends on the entire distribution, such as the Gini coefficient. While we do not claim that the Pareto exponent is the only interesting inequality measure, it is certainly a robust (detail-independent) measure for top tail inequality. See \cite{gabaix2009,Gabaix2016JEP} for more discussion on the robustness of the Pareto distribution.

In this paper, we use the harmonized \emph{Luxembourg Income Study} database (hereafter \citetalias{LIS2021}) to document the capital and labor income Pareto exponents across all available \Ncross country-year observations that span \ncountries countries over half a century. We document two empirical findings. First, we find that the capital income Pareto exponent is roughly in the range 1--3 (with median \medcap) and is smaller than the labor income Pareto exponent, which ranges between 2--5 (with median \medlab). This implies that capital income is more unequally distributed than labor income. This fact is unsurprising and well known for a specific country or year (see, for example, the Lorenz curve in Figure 1 of \citealp{SaezStancheva2018}). However, we are not aware of a comprehensive study that systematically analyzes datasets from many countries and years, and therefore our finding suggests that capital income is generally more unequal than labor income. More specifically, using a statistical test recently developed by \cite{hoga2018detecting}, we formally test the equality of capital and labor income Pareto exponents and the null is rejected in \rejectPercent of samples. In every single case of rejection, the capital exponent is smaller than the labor exponent. Second, we find that the capital income Pareto exponent is nearly unrelated to the labor exponent. In particular, the correlation between the two exponents across countries is close to zero.

To explain our empirical findings, we build a simple incomplete market model with job ladders and capital income risk. In the model, agents get randomly promoted to the next job ladder. Because individual income follows a random multiplicative process, we obtain a Pareto-tailed labor income distribution. The agents also save assets and face idiosyncratic investment risk, which generates a Pareto-tailed wealth (hence capital income) distribution. Because the capital income Pareto exponent is mainly determined by the asset return distribution, while the labor income Pareto exponent is mainly determined by the income growth distribution, the relation between the two is weak. Furthermore, we analytically characterize the capital and labor income Pareto exponents and show that the former tends to be smaller than the latter for common parametrization. Our results suggest the importance of distinguishing income and wealth inequality.

\paragraph{Related literature}

The power law behavior of income was first recognized by \cite{Pareto1895LaLegge,Pareto1896LaCourbe,Pareto1897Cours}, who used tabulation data of tax returns in many European countries. More recent research that employs micro data include \cite{reed2001} for U.S., \cite{reed2003} for U.S., Canada, Sri Lanka, and Bohemia, \cite{nirei-souma2007} for Japan, \cite{Toda2011PRE,Toda2012JEBO} for U.S., \cite{Jenkins2017} for U.K., and \cite{IbragimovIbragimov2018} for Russia. These papers all concern specific countries and years. \cite{BandourianMcDonaldTurley2003} estimate eleven parametric distributions (some of which exhibit Pareto tails) using 82 household labor income datasets from \emph{Luxembourg Income Study} \citepalias{LIS2021} as we do, though they neither focus on the Pareto exponent nor consider capital income. \cite{gabaix2009} mentions ``[the] tail exponent of income seems to vary between 1.5 and 3'', citing \cite{AtkinsonPiketty2007}, though without providing specific details. \citet[Table 13A.23]{AtkinsonPiketty2010} document income Pareto exponents across many countries and years estimated from top income share data based on tax returns. However, these estimates are computed from total income, and since (as we document in Section \ref{subsec:test}) the capital income Pareto exponents tend to be smaller than labor exponents, their estimates are best understood as capital income (hence wealth) Pareto exponents. \cite{BenhabibBisinLuo2017} make the point that wealth is more skewed than income, citing a few Pareto exponent estimates from \cite{BadelDalyHuggettNybom2018} for income and \cite{Vermeulen2018} for wealth. Using the 2013-2014 individual income tax data from Romania, \cite{Oancea_2018} document that the capital income Pareto exponent (1.44) is smaller than the labor exponent (2.53). Using time series regressions for 21 countries, \cite{Bengtsson_2018} document a positive correlation between the gross capital share in national accounts and the top 1\% income shares. Their finding can be explained if top income earners tend to have higher capital income share, which is the case if the capital income Pareto exponent is smaller than the labor exponent as we document in this paper. As mentioned in the introduction, there seems to be no comprehensive studies that document the capital and labor income Pareto exponents separately for many countries and years.

\section{Data}\label{sec:data}

In this section we describe the dataset that we use and discuss its limitations.

\subsection{The LIS database}
We use the data from the \emph{Luxembourg Income Study} \citepalias{LIS2021}, which is a large, harmonized database of micro-level income data that covers over 50 countries worldwide and many years since the late 1960s. In many countries, the data derive from government surveys (for example, the U.S.\ data is based on the \emph{Current Population Survey}). The \citetalias{LIS2021} data are available at both individual and household level. We focus on the household labor and capital income because
\begin{enumerate*}
	\item it is reasonable to assume that economic decisions such as financial planning are made at the household level, and
	\item incomes among couples are likely correlated due to assortative matching in the marriage market \citep{Siow_2015}, which invalidates statistical estimation.\footnote{In our data, we find an average correlation of 0.22 between labor income of husband and wife, which underpins the conjectured dependency.}
\end{enumerate*}
The \citetalias{LIS2021} defines \emph{labor income} as ``cash payments and value of goods and services received from dependent employment, as well as profits/losses and value of goods from self-employment, including own consumption''. \emph{Capital income} is defined as ``cash payments from property and capital (including financial and non-financial assets), including interest and dividends, rental income and royalties, and other capital income from investment in self-employment activity''. Together these two categories make up total \emph{factor income}. See the \citetalias{LIS2021} 2019 USER GUIDE\footnote{\url{https://www.lisdatacenter.org/wp-content/uploads/files/data-lis-guide.pdf}} for a detailed summary on how these data are retrieved and calculated.

\subsection{Data limitations}
Our analysis draws upon datasets from many different countries that are harmonized into a common framework by the \citetalias{LIS2021}. However, many details about the collection of data in the different countries are omitted. For example, we find evidence of top-coding in some countries and years, as the largest income order statistic is equal to the second largest.\footnote{Among all \Ncross country-year observations, the first and second order statistics are equal in 6 cases for labor income and 11 cases for capital income. Therefore we conjecture that the top-coding issue is not severe.} Top-coding induces an upward bias in the estimation of the Pareto exponent. This issue is not necessarily resolved if, instead, one relies on administrative tax income data, for similar biases arise such as rich households trying to understate their taxable income \citep{atkinson2011top}. \cite{burkhauser2012recent} detail a method that can be used to overcome the bias due to top-coding, however at the end of their paper they show that the results are robust even if estimates are based on the top-coded series. For these reasons we treat the datasets as not being top-coded in our analysis.

Another limitation of the \citetalias{LIS2021} database is that it is based on government surveys and the measurement error may be larger compared to administrative data based on tax returns. The fact that the income distribution in administrative data is often reported as tabulations, not micro data, causes no problem for estimating Pareto exponents, as \cite{TodaWang2021JAE} provide an efficient estimation method for such data. In fact, \citet[Table 13A.23]{AtkinsonPiketty2010} document income Pareto exponents across countries and years estimated from top income share data. However, their table is based on total income, and since (as we document in Section \ref{subsec:test}) the capital income Pareto exponents tend to be smaller than labor exponents, the estimates in \cite{AtkinsonPiketty2010} are best understood as capital income (hence wealth) Pareto exponents. Since we are not aware of a comprehensive income database that distinguishes capital and labor income, we decided to use the \citetalias{LIS2021} database.

\section{Pareto exponents across countries and years}\label{sec:Pareto}

In this section we estimate the capital and labor income Pareto exponents for all countries and years that are available in the \citetalias{LIS2021} database, which spans \ncountries countries over half a century (\Ncross country-year observations in total). We then formally test for the equality of the Pareto exponents of capital and labor income.

\subsection{Estimation method}\label{subsec:estim}
For each country and year, we suppose that the (capital or labor) income observations $\set{X_n}_{n=1}^N$ are independent and identically distributed (\iid) with cumulative distribution function (CDF) $F(x)=\Pr(X_n \le x)$. The assumption that the upper tail of income obeys a power law with Pareto exponent $\alpha>0$ translates into the regular variation condition
\begin{equation}
1-F(x)=x^{-\alpha}\ell(x) \label{eq:RV}
\end{equation}
for some slowly varying function $\ell$ (see Footnote \ref{fn:RV}). Note that the assumption on $\ell$ involves only the limit as $x\to\infty$; we are assuming a power law behavior in the upper tail without taking a stance on the shape of the entire distribution.

We are interested in estimating the Pareto exponent $\alpha$ for each country and year. For this purpose, we employ the \citet{hill1975simple} (maximum likelihood) estimator 
\begin{equation}
\frac{1}{\hat{\alpha}(k)} \coloneqq \frac{1}{k} \sum_{n=1}^k \log\left(\frac{X_{(n)}}{X_{(k)}} \right). \label{eq:Hill}
\end{equation}
Here $X_{(n)}$ denotes the $n$-th largest order statistic from the sample $\set{X_n}_{n=1}^N$ and $k\in \set{1,\dots,N}$ denotes the number of tail observations used to estimate the Pareto exponent. \citet{Hall1982} shows that the standard error of $\hat{\alpha}(k)$ is $\hat{\alpha}(k)/\sqrt{k}$ under flexible assumptions on the CDF (see \eqref{eq:hall} below). We use the Hill estimator because we are interested in formally testing the equality of the capital and labor income Pareto exponents using the method of \cite{hoga2018detecting}, for which the Hill estimator is required.\footnote{If we are only interested in estimating the Pareto exponents, then there are many alternative methods available. \cite{GomesGuillou2015} review 13 commonly used estimators. \cite{Fedotenkov2020} reviews more than 100.}

When the population distribution is known to be exactly Pareto (so $\ell$ in \eqref{eq:RV} is zero below the minimum size $x_{\min}$ and constant above this threshold), it is well known that the Hill estimator for the full sample ($k=N$) is consistent, asymptotically normal, and asymptotically efficient because it is a maximum likelihood estimator. In practice, the CDF is not exactly Pareto and the researcher needs to select an appropriate value of $k$. For instance, if $F(x)$ satisfies
\begin{equation}\label{eq:hall}
1-F(x) = Cx^{-\alpha} (1 + Dx^{-\beta} + o(x^{-\beta}))
\end{equation}
with some $\beta>0$, then \cite{Hall1982} shows that choosing $k = o(N^{2\beta/(2\beta + \alpha)})$ together with $k \to \infty$ as $N \to \infty$ is sufficient for consistency and asymptotic normality (see also \citealp{embrechts2013modelling}).\footnote{Recall that we write $f(x)=o(g(x))$ if for all $\epsilon>0$, we have $\abs{f(x)}\le \epsilon\abs{g(x)}$ for large enough $x$.} Notice that this choice puts a bound on the growth rate of $k$. The expansion \eqref{eq:hall} covers a wide range of distributions of interest, such as the $t$-distribution and the type II extreme value distribution \citep{danielsson1997tail}.

Despite these asymptotic results, it is notoriously difficult to pick $k$ optimally in finite samples \citep{Hall1990,ResnickStarica1997,Danielsson_2001}. In practice, researchers often plot the Hill estimator \eqref{eq:Hill} over a range of $k$ to find a flat region or plot the log rank $\log 1,\dots,\log N$ against the log size $\log X_{(1)},\dots,\log X_{(N)}$ to find a region that exhibits a straight line pattern and choose a size threshold to run the log-rank regression.\footnote{\cite{gabaix-ibragimov2011} study the asymptotic behavior of log rank regression and show that the standard error is larger by a factor of $\sqrt{2}$ than the Hill estimator. However, they do not discuss how to select the threshold. In their empirical application, they consider the size distribution of population in U.S.\ metropolitan statistical areas, which are already far into the tail and hence the threshold selection is less of an issue. \cite{IbragimovIbragimov2018} apply the same methodology to Russian household income data and consider the top 5\% and 10\% thresholds.} Unfortunately, this graphical approach is not feasible in our setting because \citetalias{LIS2021} does not allow researchers to download the micro data for confidentiality concerns (researchers are required to submit their execution files to conduct statistical analyses) and there is little scope for exploratory graphical data analysis. In this paper we simply use the largest 5\% observations, so 
\begin{equation}\label{eq:5prule}
k=\floor{0.05N},
\end{equation}
which is standard in the literature.\footnote{An alternative approach is to estimate a parametric distribution $F$ that admits a Pareto upper tail by maximum likelihood using the entire sample. The double Pareto-lognormal distribution proposed by \cite{reed2003} and \cite{reed-jorgensen2004} often performs best. See \cite{Toda2012JEBO} for a horse race across several parametric distributions in the context of U.S.\ labor income.} Unreported simulations show that our results are robust to using other thresholds such as the largest 1\% and 10\% observations, or using a data-driven procedure using the Kolmogorov-Smirnov distance as in \cite{danielsson2016tail}.

\subsection{Capital and labor income Pareto exponents}\label{subsec:results}

We estimate the capital and labor income Pareto exponents for all countries and years available in the \citetalias{LIS2021} database. The database spans \ncountries countries across the years 1967--2018, with a total of \Ncross country-year observations. The point estimates of the capital and labor income Pareto exponents for each country and year as well as their standard errors can be found in Table \ref{t:Pareto} in Appendix \ref{sec:tables}. To avoid small sample issues, we restrict our analysis to countries with at least \Nmin positive observations for income, resulting in (all) \Ncross country-year pairs for labor income and \Ncap for capital income. For visibility, Figure \ref{fig:alpha} shows the histogram and scatter plot of the capital and labor income Pareto exponents.

\begin{figure}[!htb]
\centering
\begin{subfigure}{0.48\linewidth}
\includegraphics[width=\linewidth]{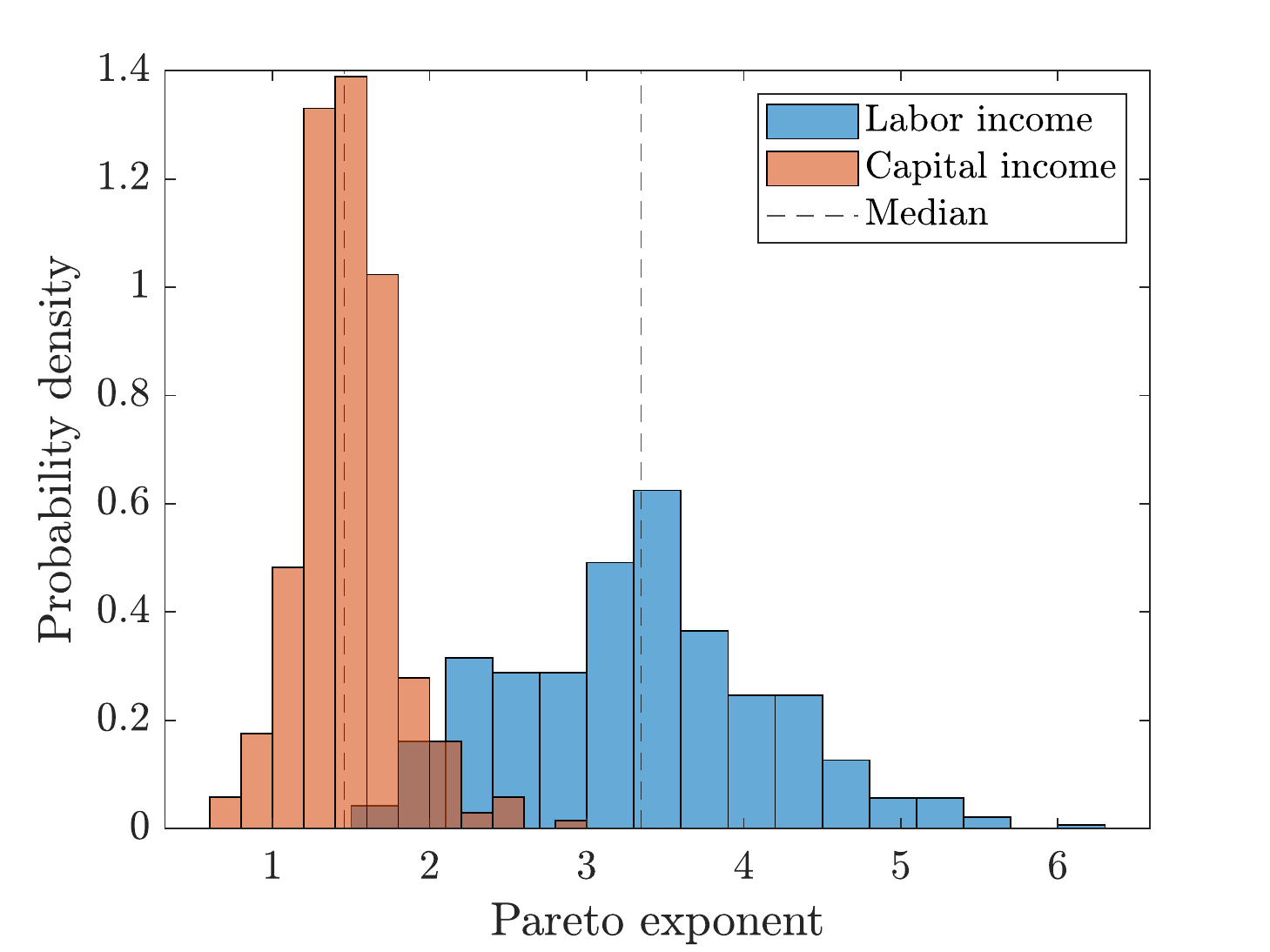}
\caption{Histogram.}\label{fig:alphaHist}
\end{subfigure}
\begin{subfigure}{0.48\linewidth}
\includegraphics[width=\linewidth]{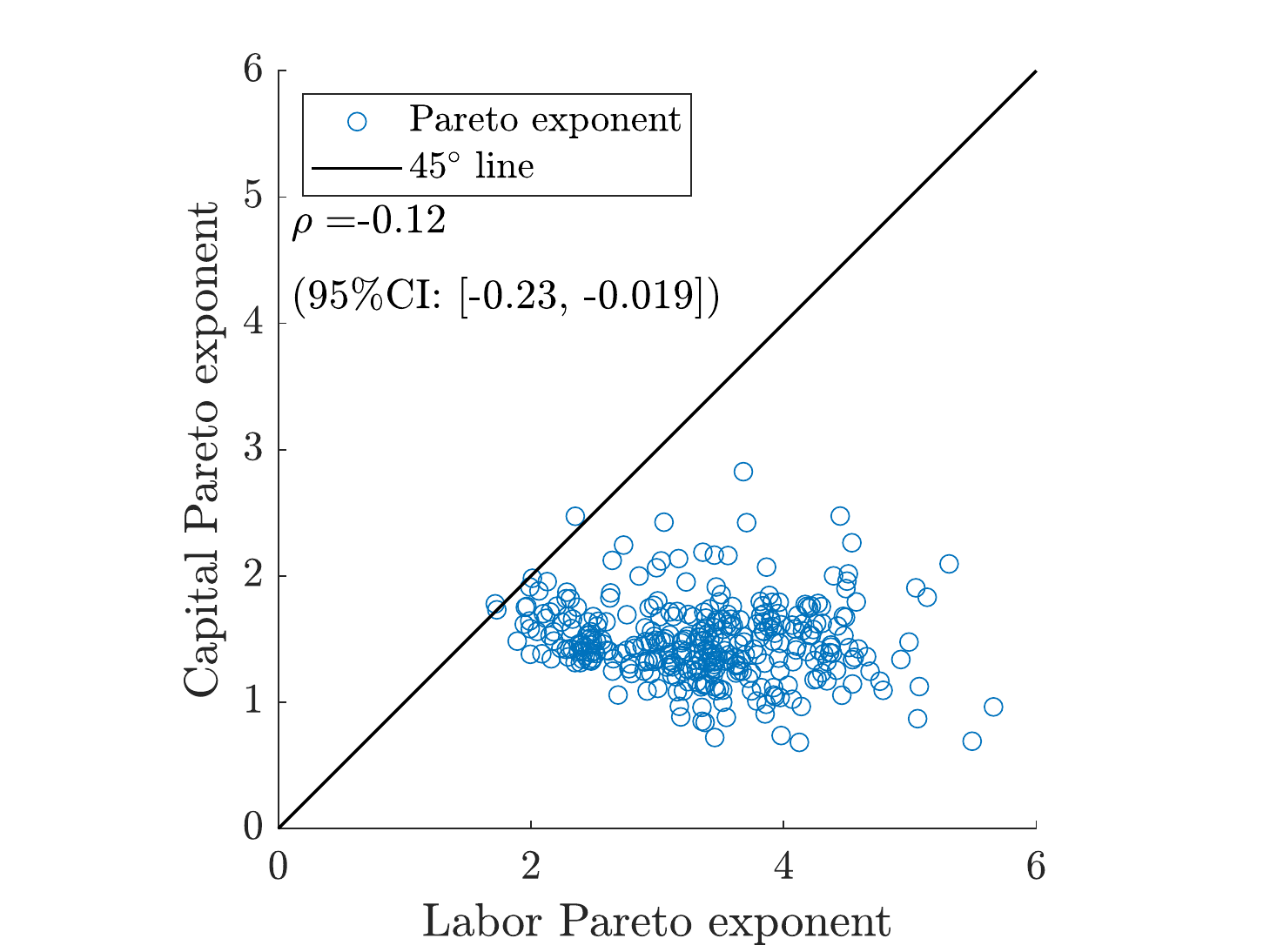}
\caption{Scatter plot.}\label{fig:alphaScatter}
\end{subfigure}
\caption{Histogram and scatter plot of capital and labor income Pareto exponents.}\label{fig:alpha}
\end{figure}

Figure \ref{fig:alphaHist} shows the histogram of the capital and labor income Pareto exponents pooled across all available countries and years. The capital and labor income Pareto exponents are generally in the range 1--3 and 2--5 with medians \medcap and \medlab, respectively. This suggests that
\begin{enumerate*}
\item capital income is generally more unequally distributed than labor income, but
\item there is significant heterogeneity in both capital and labor income inequality across countries and years.
\end{enumerate*}
Figure \ref{fig:alphaScatter} shows the scatter plot of the Pareto exponents together with the 45 degree line. The confidence interval of the correlation coefficient $\rho$ is computed assuming all country-year observations are independent and there is no sampling error in the point estimates of the Pareto exponents. Although it is not obvious how to account for these issues, doing so will only widen the confidence interval. Therefore the fact that the na\"{\i}ve confidence interval almost contains zero suggests that the correlation between the two Pareto exponents is weak. Furthermore, for the vast majority of countries and years, the capital income Pareto exponent is smaller than the labor exponent, again suggesting that capital income is more unequal than labor income.

How do the Pareto exponents evolve over time? Because many countries appear only sporadically in the \citetalias{LIS2021} database, we only consider six countries for which the cross-sectional sample size is large and the time series is long, namely Canada, Germany, Switzerland, Taiwan, United Kingdom, and United States. Figure \ref{fig:alphaTSselect} plots the capital and labor Pareto exponents and their 95\% confidence intervals for these countries. The capital exponent tends to be stable at around 1--2 and is smaller than the labor exponent.\footnote{An exception may be Germany before 1983. However, this may be an artifact of the sudden change in sample size, which was over 40,000 until 1983 and about 5,000 since 1984. Thus the 5\% rule for capital income may be including too many observations in the body of the distribution before 1983.} Again, capital income appears to be more unequally distributed than the labor income.

\begin{figure}[!htb]
\centering
\includegraphics[width=0.48\linewidth]{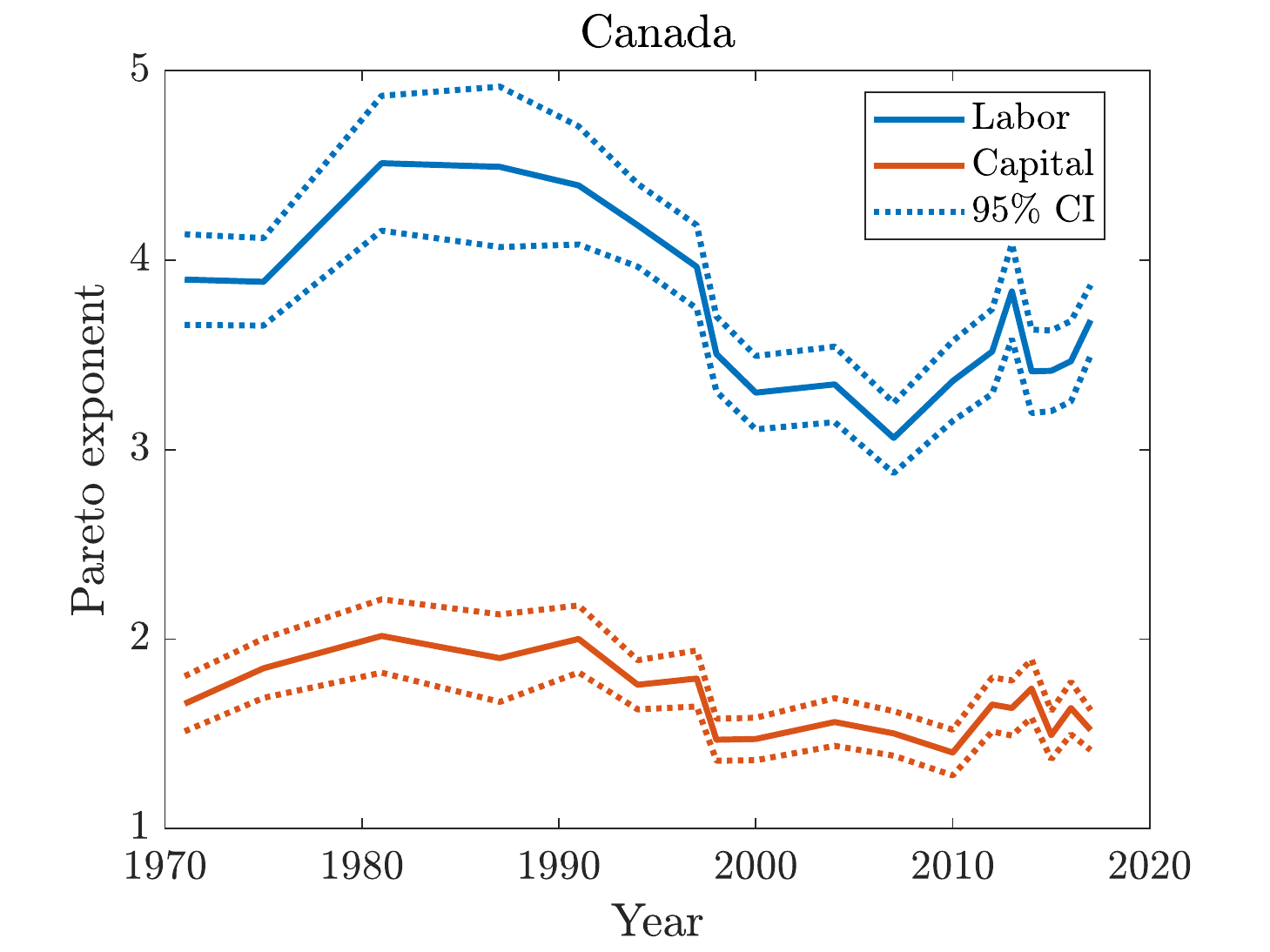}
\includegraphics[width=0.48\linewidth]{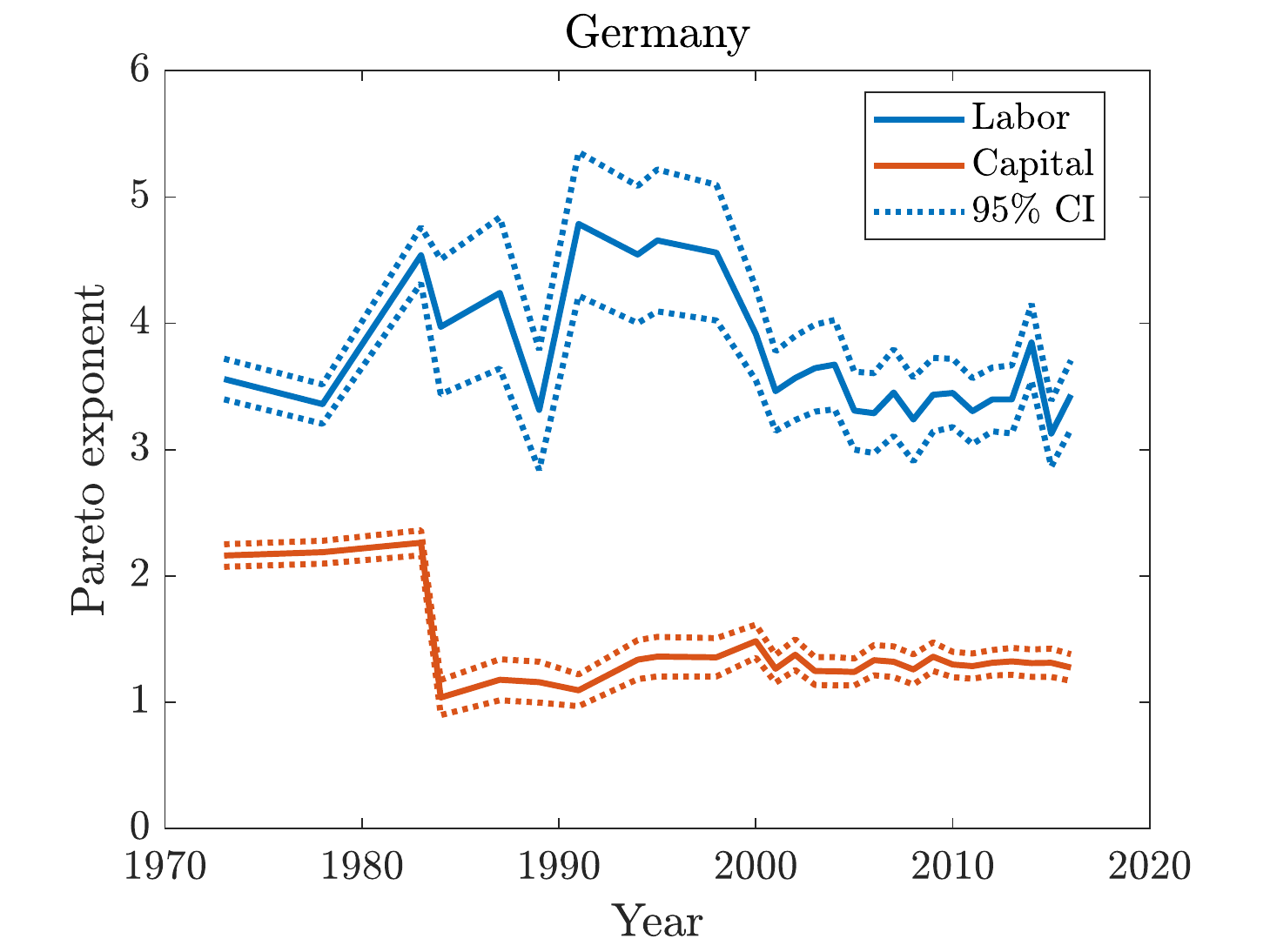}
\includegraphics[width=0.48\linewidth]{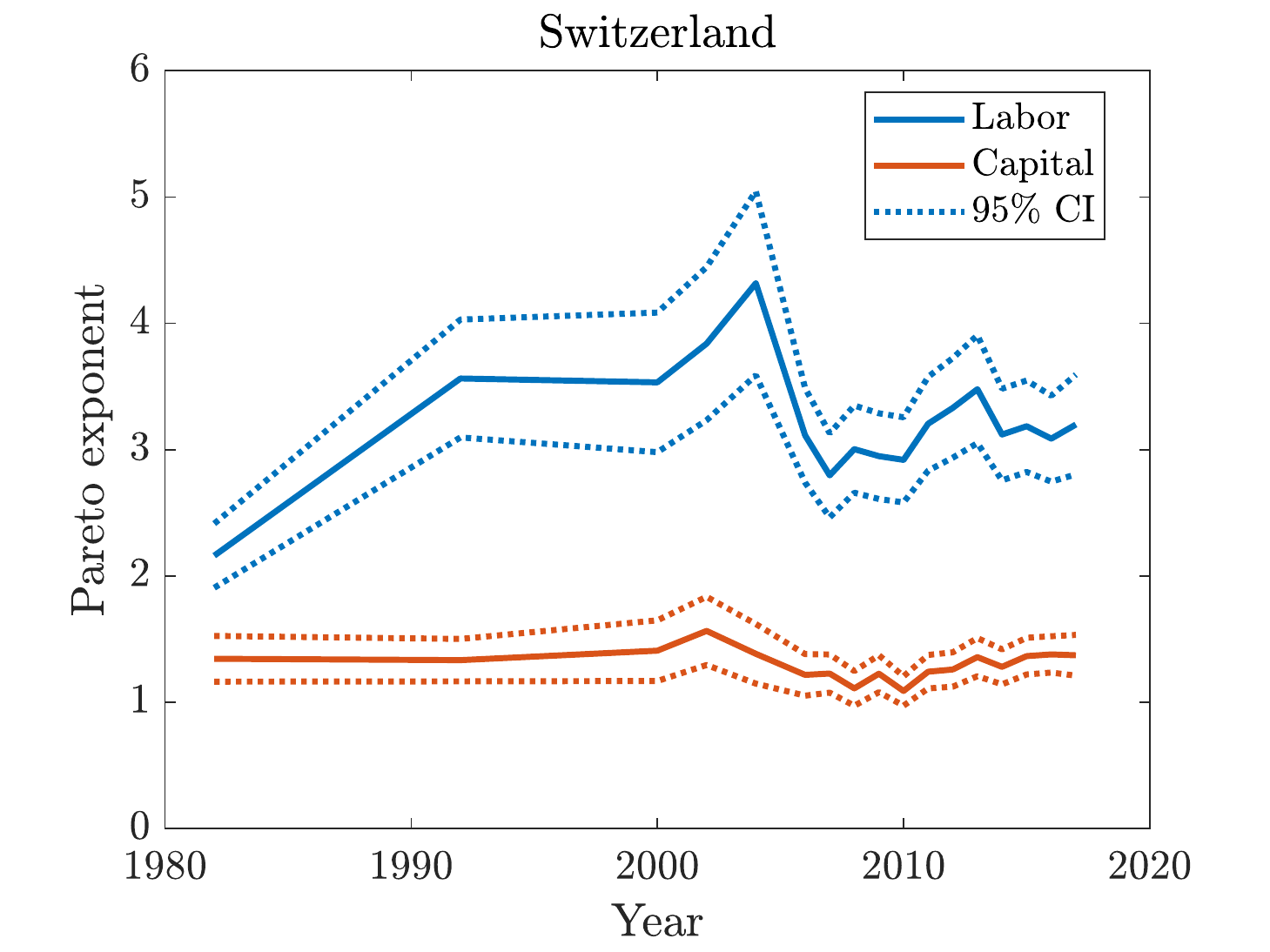}
\includegraphics[width=0.48\linewidth]{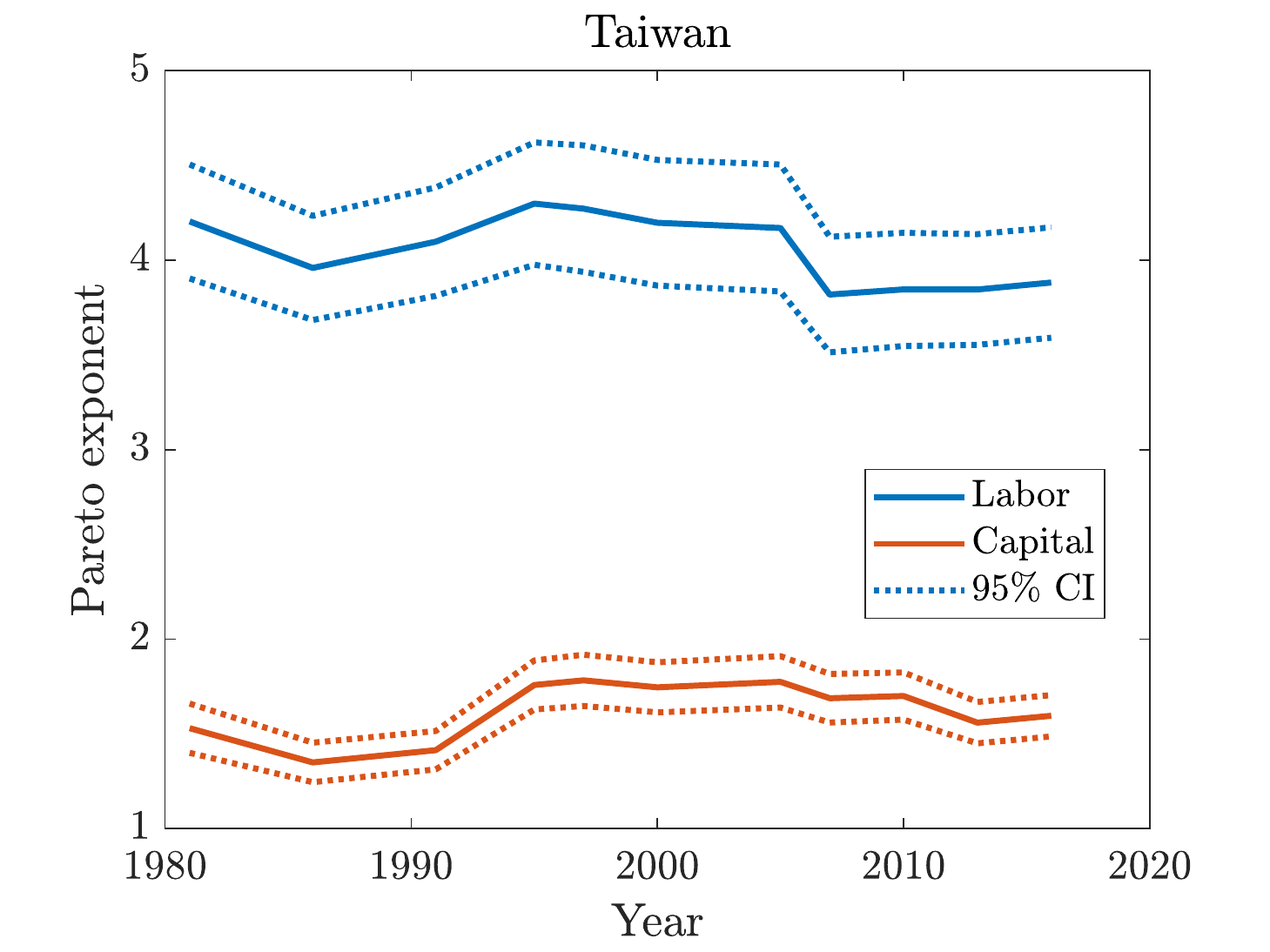}
\includegraphics[width=0.48\linewidth]{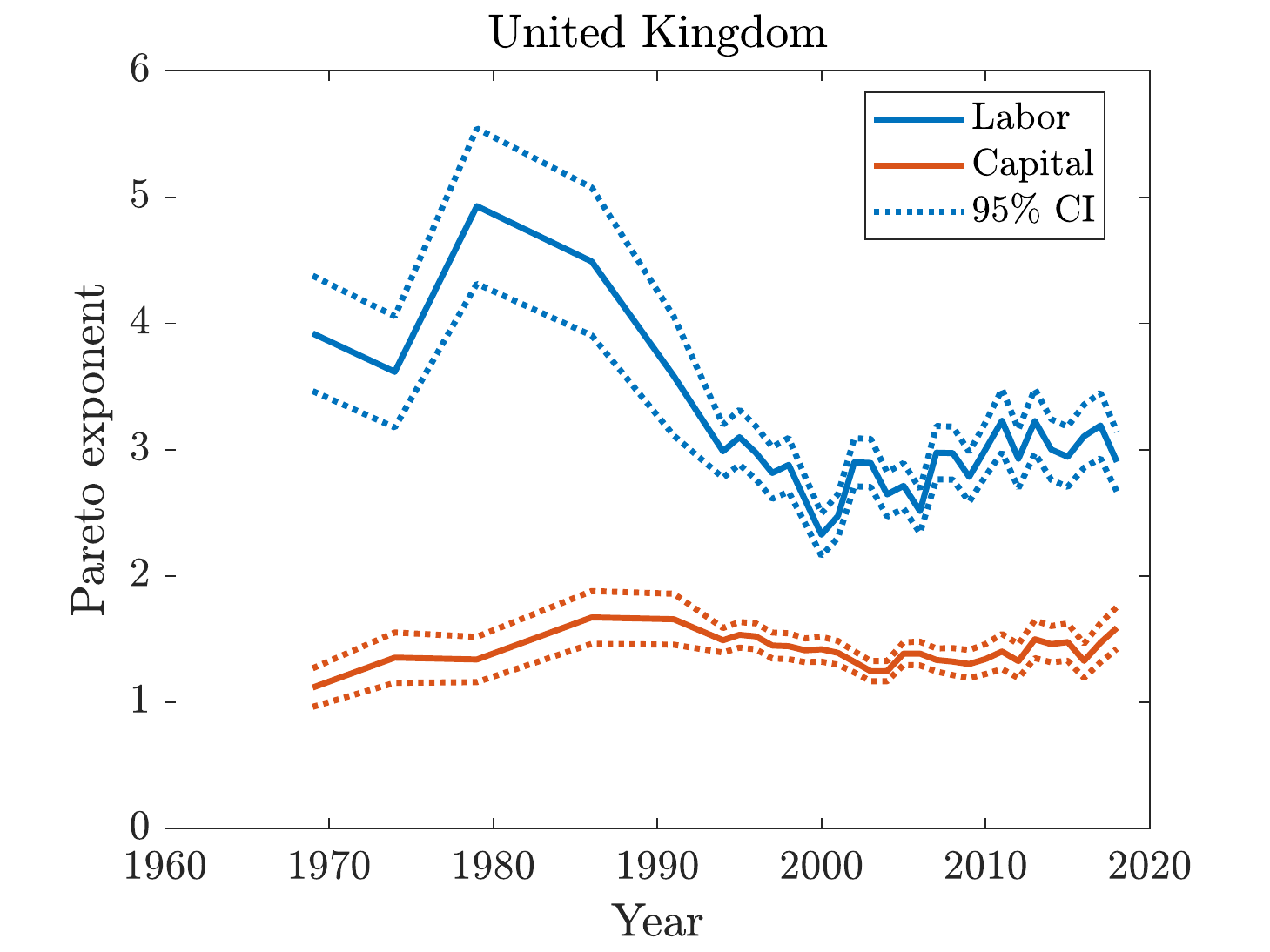}
\includegraphics[width=0.48\linewidth]{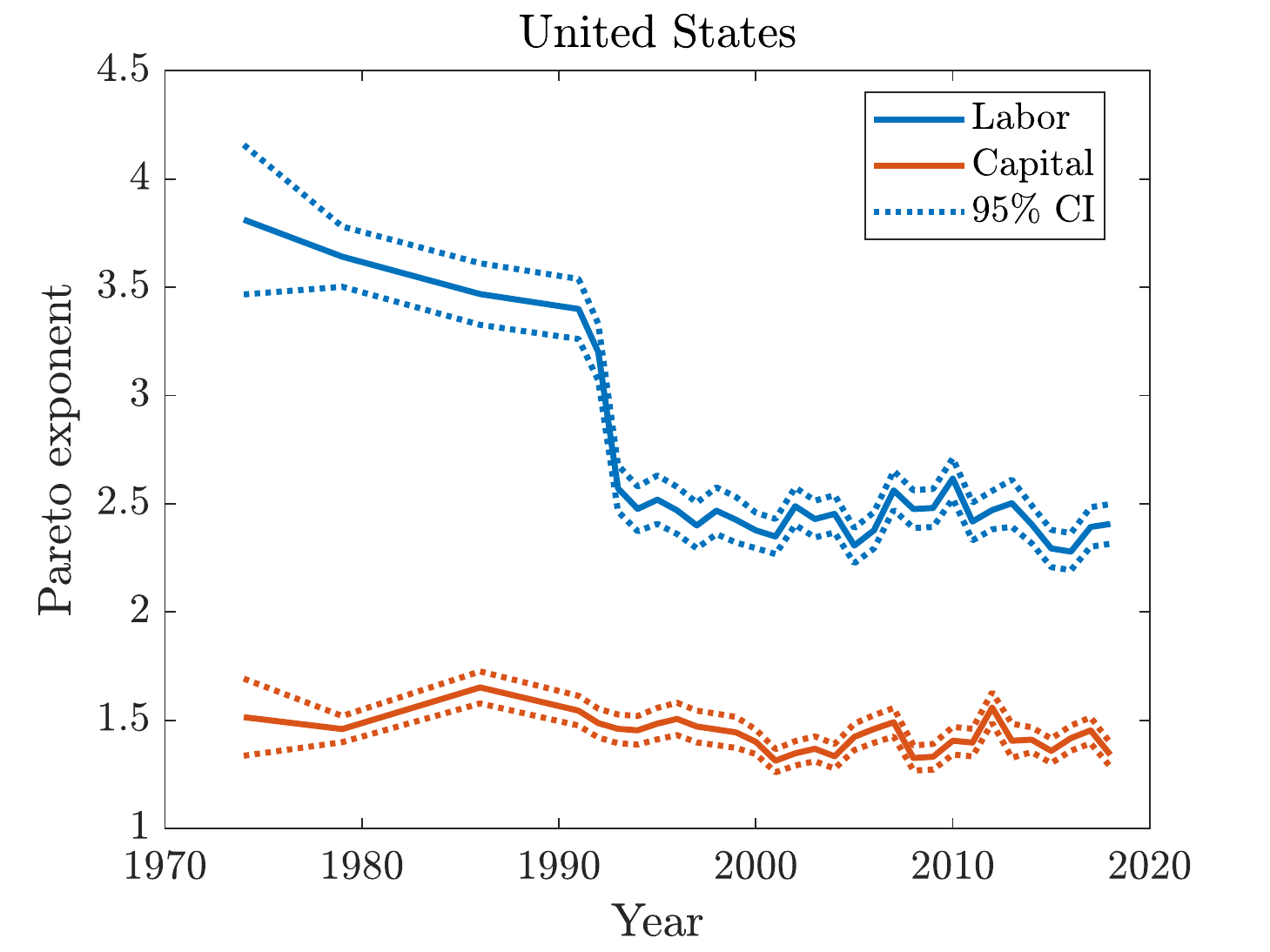}
\caption{Time evolution of capital and labor income Pareto exponents.}\label{fig:alphaTSselect}
\end{figure}

\subsection{Testing equality of capital and labor Pareto exponents}\label{subsec:test}

We now formally test whether the capital and labor Pareto exponents are equal. In particular our test is
\begin{equation*}
H_0: \alpha_\labor = \alpha_\capital \quad \text{against} \quad H_1: \alpha_\labor \ne \alpha_\capital,
\end{equation*}
where $\alpha_\capital,\alpha_\labor$ denote the capital and labor income Pareto exponents. Testing the null hypothesis $H_0$ is complicated by the fact that there is dependency between labor income $\set{X_{\labor,n}}_{n=1}^N$ and capital income $\set{X_{\capital,n}}_{n=1}^N$, because individuals who are rich (receive high labor income) tend to be wealthy and receive high capital income. Thus we cannot use the 95\% confidence intervals in Figure \ref{fig:alphaTSselect} to test the equality of Pareto exponents. Instead, we apply the test recently developed by \citet{hoga2018detecting}, which allows for  dependence in the data but assumes restrictions on the growth rate of the tail dependence (see \citealp[Assumption A2]{hoga2018detecting}). The test is based upon the inverse of the Hill estimator \eqref{eq:Hill}, which we denote by $\hat{\gamma} \coloneqq 1/\hat{\alpha}$. The test statistic is defined by
\begin{equation}
T_N = \frac{(\hat{\gamma}_\labor(1) - \hat{\gamma}_\capital(1))^2}{\int_{t_0}^{1}t^2 \left[(\hat{\gamma}_\labor(t) - \hat{\gamma}_\capital(t)) - (\hat{\gamma}_\labor(1) - \hat{\gamma}_\capital(1))\right]^2 \diff t}, \label{eq:HogaT}
\end{equation}
where $t_0 \in (0,1)$ is a tuning parameter and $\hat{\gamma}(t)$ is the inverse Hill estimator 
\begin{equation}
\hat{\gamma}(t) \coloneqq \frac{1}{\floor{kt}} \sum_{n = 1}^{\floor{kt}} \log\left(\frac{X_{(n)}}{X_{(\floor{kt})}}\right). \label{eq:gamhat}
\end{equation}
Using the Hill estimator based on the subsample with only $\floor{kt}$ observations leads to self-normalization of the test statistic $T_N$ and renders a test that is asymptotically pivotal. The limiting distribution is
\begin{equation}
T_N \dto \frac{W(1)^2}{\int_{t_0}^1 [W(t) - tW(1)]^2 \diff t}, \label{eq:HogaTlim}
\end{equation}
where $W(t)$ is a standard Brownian motion. Since the test statistic \eqref{eq:HogaT} can be computed using only the Hill estimator and conducting numerical integration, there is no need to estimate the (potentially difficult) tail covariance. The tuning parameter $t_0$ affects the size of the test in finite samples: high values of $t_0$ make the integral in \eqref{eq:HogaT} based on too few differences of $\hat{\gamma}$, and low values of $t_0$ yield volatile $\hat{\gamma}$ in \eqref{eq:gamhat} when $t$ is close to $t_0$. Both of these effects may cause size distortions. Therefore we set $t_0 = 0.2$ following the recommendation of \cite{hoga2018detecting}, who finds that this choice leads to favorable size properties.\footnote{The choice of $t_0$ in \citet{hoga2018detecting} comes out using an automated selection procedure to choose $k$, which is different than our 5\% rule. An earlier version of our paper employs the same automated procedure, which leads to very similar results.} We reject the null $H_0$ when the test statistic $T_N$ is large. According to Table I of \cite{hoga2018detecting}, the 95 percentile of \eqref{eq:HogaTlim} for $t_0=0.2$ is 55.44, which we use as the critical value for testing $H_0$ at 5\% significance level. 

One issue with the test statistic \eqref{eq:HogaT} is that it requires the same number of tail observations $k$ for both cross-sections of capital and labor income. Hence the 5\% rule \eqref{eq:5prule} discussed in Section \ref{subsec:estim} becomes problematic as the number of people with capital income in our data set is rather small. Many households do not hold liquid financial wealth and hence have no capital income. The resulting test is thus not feasible since $k_\labor$ based on our 5\% rule could be wildly different from $k_\capital$. To overcome this issue, we only test the equality of Pareto exponents for countries that have more than \Nmin positive capital income observations and set $k = \floor{0.05N_\capital}$, where $N_\capital$ is the number of positive capital income observations (these households always have positive labor income), resulting in \Ncap country-year observations out of \Ncross. In practice this means that we estimate the Pareto exponent of labor income further in the tail because the sample size of labor income $N_\labor$ tends to be larger than $N_\capital$. However, this is acceptable since the Pareto approximation tends to fit better for smaller $k$. Figure \ref{fig:alphaLScatter} presents a scatter plot of the estimated labor income Pareto exponent $\hat{\alpha}_\labor$ using 5\% of the full sample ($N_\labor$) and 5\% of the sample with positive capital income ($N_\capital$). The fact that most points are close to the 45 degree supports our claim.

\begin{figure}[!htb]
\centering
\includegraphics[width=0.7\linewidth]{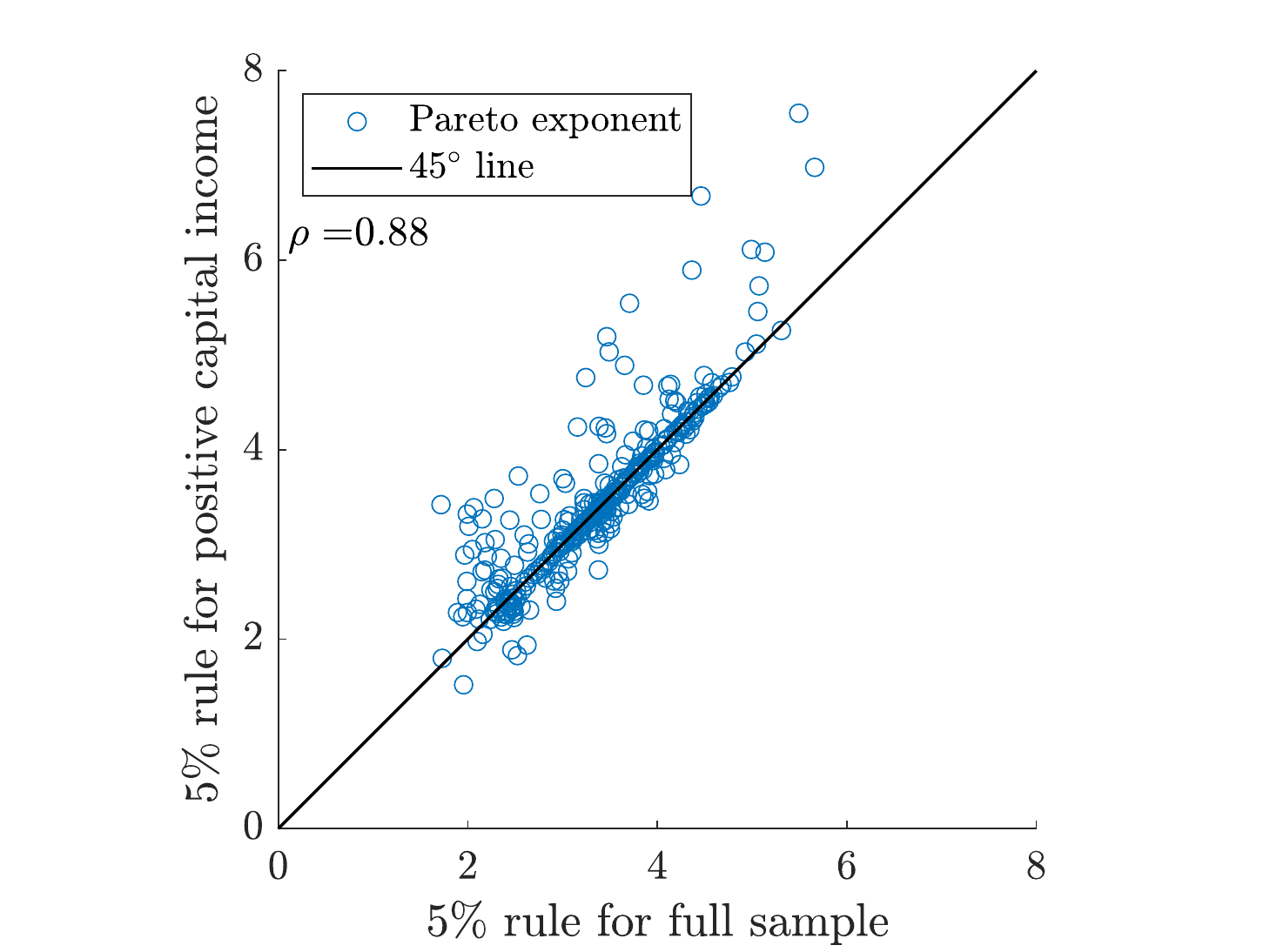}
\caption{Labor Pareto exponent with full and positive capital income samples.}\label{fig:alphaLScatter}
\end{figure}

Table \ref{t:equal} in Appendix \ref{sec:tables} shows the test results of the null hypothesis $H_0: \alpha_\labor = \alpha_\capital$. 
We reject the null in \Nreject country-year observations out of \Ncap (\rejectPercent) that meet our sample selection criterion. In every single case of rejection, we have $\hat{\alpha}_\labor>\hat{\alpha}_\capital$, and therefore we formally confirm the observation in Section \ref{subsec:results} that capital income is more unequally distributed than labor income.

\section{Model of capital and labor Pareto exponents}\label{sec:IFP}

Our empirical analysis in Section \ref{sec:Pareto} suggests that
\begin{enumerate*}
\item the capital income Pareto exponent is smaller than the labor one (\ie, capital income is more unequally distributed than labor income), and
\item the correlation between capital and labor income Pareto exponents is weak.
\end{enumerate*}
To explain these empirical findings, we present a simple dynamic model of consumption and savings, which builds on one of the authors' prior works \citep{MaStachurskiToda2020JET,MaToda2021JET}. Our model is more specialized but the characterizations are sharper. The proofs of propositions in this section are deferred to Appendix \ref{sec:proof}.

\subsection{Income fluctuation problem}\label{subsec:IF}
Time is discrete and denoted by $t=0,1,2,\dotsc$. Let $a_t$ be the financial wealth of a typical agent at the beginning of period $t$ including current income. The agent chooses consumption $c_t\ge 0$ and saves the remaining wealth $a_t-c_t$. The period utility function is $u:(0,\infty)\to \R$, the discount factor is $\beta>0$, the gross return on wealth between time $t-1$ and $t$ is $R_t>0$, and non-financial income at time $t$ is $Y_t>0$. Thus the agent solves
\begin{subequations}\label{eq:IF}
\begin{align}
&\maximize && \E_0\sum_{t=0}^\infty \beta^tu(c_t) \label{eq:IF_utility}\\
&\st && a_{t+1}=R_{t+1}(a_t-c_t)+Y_{t+1}, \label{eq:IF_budget}\\
&&& 0\le c_t\le a_t, \label{eq:IF_borrow}
\end{align}
\end{subequations}
where the initial wealth $a_0=a>0$ is given, \eqref{eq:IF_budget} is the budget constraint, and \eqref{eq:IF_borrow} implies that the agent cannot borrow (which is without loss of generality according to the discussion in \citealp{ChamberlainWilson2000}). Throughout the rest of the paper we maintain the following assumptions.

\begin{asmp}[CRRA utility]\label{asmp:CRRA}
The utility function exhibits constant relative risk aversion (CRRA) with coefficient $\gamma>0$, so $u(c)=\frac{c^{1-\gamma}}{1-\gamma}$ if $\gamma\neq 1$ and $u(c)=\log c$ if $\gamma=1$.
\end{asmp}

\begin{asmp}[\iid shocks]\label{asmp:iid}
Let $G_{t+1}\coloneqq Y_{t+1}/Y_t$ be gross growth rate of income. The sequence $\set{R_{t+1},G_{t+1}}_{t=0}^\infty$ is independent and identically distributed (\iid).
\end{asmp}

These assumptions are similar to \cite{Carroll2020}, except that we allow for stochastic returns on savings. Note that the asset return $R_{t+1}$ and income growth $G_{t+1}$ are potentially mutually dependent. Due to the \iid assumption, the state variables of the income fluctuation problem \eqref{eq:IF} are financial wealth $a_t>0$ and current income $Y_t>0$. Exploiting homotheticity (Assumption \ref{asmp:CRRA}), we can reduce the number of state variables to just one, namely the wealth-income ratio (normalized wealth) $\tilde{a}_t\coloneqq a_t/Y_t$. To see this, letting $\tilde{c}_t\coloneqq c_t/Y_t$ be the consumption-income ratio (normalized consumption), dividing the borrowing constraint \eqref{eq:IF_borrow} by $Y_t$, we obtain $0\le \tilde{c}_t\le \tilde{a}_t$. Similarly, dividing the budget constraint \eqref{eq:IF_budget} by $Y_{t+1}$, we obtain
\begin{align}
\tilde{a}_{t+1}=a_{t+1}/Y_{t+1}&=(R_{t+1}Y_t/Y_{t+1})(a_t/Y_t-c_t/Y_t)+1 \notag \\
&=(R_{t+1}/G_{t+1})(\tilde{a}_t-\tilde{c}_t)+1 \notag \\
&=\tilde{R}_{t+1}(\tilde{a}_t-\tilde{c}_t)+1, \label{eq:detrend_budget}
\end{align}
where $\tilde{R}_{t+1}\coloneqq R_{t+1}/G_{t+1}$ is the asset return relative to income growth. As for the utility function, since
\begin{equation*}
c_t=Y_t\tilde{c}_t=Y_0\left(\prod_{s=1}^tG_s\right)\tilde{c}_t
\end{equation*}
(here we interpret $\prod_{s=1}^0\bullet=1$), assuming $Y_0=1$ (which is without loss of generality) and $\gamma\neq 1$, it follows from \eqref{eq:IF_utility} that
\begin{align}
\E_0\sum_{t=0}^\infty \beta^tu(c_t)&=\E_0\sum_{t=0}^\infty \left(\prod_{s=1}^t \beta G_s^{1-\gamma}\right)\frac{\tilde{c}_t^{1-\gamma}}{1-\gamma} \notag \\
&=\E_0\sum_{t=0}^\infty \left(\prod_{s=1}^t \tilde{\beta}_s\right)\frac{\tilde{c}_t^{1-\gamma}}{1-\gamma}, \label{eq:detrend_utility}
\end{align}
where $\tilde{\beta}_t\coloneqq \beta G_t^{1-\gamma}$. The discussion for $\gamma=1$ is similar. Therefore the problem reduces to an income fluctuation problem with CRRA utility, random discount factors $\set{\tilde{\beta}_t}_{t=1}^\infty$, stochastic returns $\set{\tilde{R}_t}_{t=1}^\infty$ on wealth, and constant income ($\tilde{Y}_t\equiv 1$). The general theory of income fluctuation problems with stochastic discounting, returns, and income in a Markovian setting was developed by \cite{MaStachurskiToda2020JET}. Therefore we immediately obtain the following result. In what follows, we drop the time subscript when no confusion arises.

\begin{prop}\label{prop:IF}
Suppose Assumptions \ref{asmp:CRRA}, \ref{asmp:iid} hold and
\begin{equation}
\beta\E[ G^{1-\gamma}]<1 \quad \text{and} \quad  \beta\E [RG^{-\gamma}]<1.\label{eq:spectral}
\end{equation}
Then the income fluctuation problem \eqref{eq:IF} has a unique solution. The consumption function can be expressed as $c(a,Y)=Y\tilde{c}(a/Y)$, where $\tilde{c}:(0,\infty)\to (0,\infty)$ is the consumption function of the detrended problem (maximizing \eqref{eq:detrend_utility} subject to \eqref{eq:detrend_budget}), which can be computed by policy function iteration.\footnote{See \cite{LiStachurski2014} and \cite{MaStachurskiToda2020JET} for details on policy function iteration.}
\end{prop}

\subsection{Tail behavior of income and wealth}\label{subsec:tailbehavior}

We now characterize the tail behavior of income and wealth in the context of the income fluctuation problem in Section \ref{subsec:IF}.

To make the model stationary, suppose that agents survive to the next period with probability $v\in (0,1)$ (perpetual youth model as in \citealp{yaari1965}). Whenever agents die, they are replaced by newborn agents. For simplicity, assume that the discount factor $\beta$ in \eqref{eq:IF_utility} already accounts for survival probability and that there is no market for life insurance (allowing for life insurance only changes $R$ to $R/v$ and is thus mathematically equivalent after reparametrization). Without loss of generality, suppose that newborn agents start with income $Y_0=1$. Then the income of a randomly selected agent is $Y_T$, where $T$ is a geometric random variable with mean $\frac{1}{1-v}$. By the assumption on income growth, the log income of a randomly selected agent
\begin{equation*}
\log Y_T=\log (Y_T/Y_0)=\sum_{t=1}^T\log G_t
\end{equation*}
is a geometric sum of \iid random variables, for which we can characterize the tail behavior as follows.

\begin{prop}[Income Pareto exponent]\label{prop:Pareto_income}
Suppose that $\Pr(G>1)>0$ and $1<v\E [G^z]<\infty$ for some $z>0$. Then the cross-sectional income distribution has a Pareto upper tail, whose exponent is the unique positive solution $z=\alpha_Y$ to
\begin{equation}
v\E [G^z]=1.\label{eq:BeareToda}
\end{equation}
\end{prop}
\begin{proof}
See \citet[Theorem 3.4]{BeareToda-dPL}.
\end{proof}

To characterize the tail behavior of wealth, we first note that the normalized consumption function $\tilde{c}$ in Proposition \ref{prop:IF} is concave and asymptotically linear with a specific slope.\footnote{\cite{CarrollKimball1996} showed the sufficiency of hyperbolic absolute risk aversion (HARA, which includes CRRA) for the concavity of the consumption function. \cite{Toda2021JME} proved the necessity.}

\begin{prop}[Concavity and asymptotic linearity]\label{prop:linear}
Let everything be as in Proposition \ref{prop:IF}. Then $\tilde{c}$ is concave and
\begin{equation}
\lim_{a\to\infty}\frac{\tilde{c}(a)}{a}=\begin{cases*}
1-(\E[\beta R^{1-\gamma}])^{1/\gamma} & if $\E [\beta R^{1-\gamma}]<1$,\\
0 & otherwise.
\end{cases*}
\label{eq:linear}
\end{equation}
\end{prop}

Using Proposition \ref{prop:linear} and setting $\rho=\min\set{(\E [\beta R^{1-\gamma}])^{1/\gamma},1}$, for high enough asset level, the detrended budget constraint \eqref{eq:detrend_budget} becomes approximately
\begin{equation*}
\tilde{a}_{t+1}\approx \rho\tilde{R}_{t+1}\tilde{a}_t+1,
\end{equation*}
which is a random multiplicative process (\cite{kesten1973} process). Under specific assumptions, \citet[Theorem 3.3]{MaStachurskiToda2020JET} prove that the upper tail of the stationary distribution of normalized wealth $\tilde{a}_t$ has a Pareto lower bound. Although a sharp characterization of the tail behavior is generally difficult, in our setting it is possible to obtain an exact characterization due to concavity and the \iid assumption.

\begin{prop}\label{prop:Pareto_capital}
Let $\rho=\min\set{(\E [\beta R^{1-\gamma}])^{1/\gamma},1}$ and $H=\rho\tilde{R}$. Suppose that
\begin{enumerate*}
\item $R$ is thin-tailed (meaning $\E[R^z]<\infty$ for all $z>0$),
\item $\log H$ is non-lattice (not supported on an evenly spaced grid), and
\item $\Pr(H>1)>0$, and $1<v\E [H^z]<\infty$ for some $z>0$.
\end{enumerate*}
Then the cross-sectional normalized wealth distribution is either bounded or has a Pareto upper tail, in which case the exponent is the unique positive solution $z=\tilde{\alpha}$ of
\begin{equation}
v\E [H^z]=1.\label{eq:BeareToda2}
\end{equation}
The Pareto exponent for wealth and capital income is then $\alpha=\min\set{\tilde{\alpha},\alpha_Y}$.
\end{prop}

Proposition \ref{prop:Pareto_capital} is significant despite its simplicity. According to the model, we always have $\alpha\le \alpha_Y$, typically with a strict inequality as we see in the numerical example below. This is in sharp contrast to canonical incomplete market general equilibrium models such as \cite{aiyagari1994}, where agents can save using only a risk-free asset. In such models, the impossibility theorem of \cite{StachurskiToda2019JET} implies that the tail behavior of income and wealth is the same, implying $\alpha=\alpha_Y$ in our setting.\footnote{The original proof in \cite{StachurskiToda2019JET} contained an error; it has been corrected in \cite{StachurskiToda2020Corrigendum}.} Therefore, unlike canonical incomplete market models, our model can explain the empirical fact that capital income is more unequal than labor income. The key assumption leading to this conclusion is the presence of stochastic returns.

We discuss an analytically solvable example to build intuition.

\begin{exmp}\label{exmp:IFP}
Let $\Delta>0$ be the length of time of one period and the discount factor be $\beta=\e^{-\delta\Delta}$, where $\delta>0$ is the discount rate. Suppose income grows at a constant rate $g>0$, so $G=\e^{g\Delta}$. Suppose asset return is risk-free, so $R=\e^{r \Delta}$ with $r>0$. Finally, let the survival probability be $v=\e^{-\eta\Delta}$, where $\eta$ is the death rate. Then \eqref{eq:BeareToda} becomes
\begin{equation*}
1=\e^{-\eta\Delta}\e^{zg\Delta}\iff z=\eta/g,
\end{equation*}
so the income Pareto exponent is $\alpha_Y=\eta/g$. (This is the classical result of \cite{WoldWhittle1957} in discrete-time.) Suppose in addition that $-\eta+r(1-\gamma)<0$ so that $\beta R^{1-\gamma}<1$. Since
\begin{equation*}
H=(\E [\beta R^{1-\gamma}])^{1/\gamma}\tilde{R}=(\beta R)^{1/\gamma}/G=\e^{(\frac{r-\eta}{\gamma}-g)\Delta},
\end{equation*}
solving \eqref{eq:BeareToda2} the normalized wealth Pareto exponent is
\begin{equation*}
\tilde{\alpha}=\frac{\eta\gamma}{r-\eta-g\gamma}
\end{equation*}
assuming $r-\eta-g\gamma>0$. Therefore
\begin{equation*}
\tilde{\alpha}<\alpha_Y\iff \frac{\eta\gamma}{r-\eta-g\gamma}<\frac{\eta}{g}\iff r>\eta+2g\gamma,
\end{equation*}
so the wealth (hence capital income) Pareto exponent is smaller than the labor income Pareto exponent if the return on wealth $r$ is sufficiently large. In summary, we obtain the following result: suppose $-\eta+r(1-\gamma)<0$ and let $\alpha_\capital,\alpha_\labor$ be the capital and labor income Pareto exponents. Then
\begin{equation}
\begin{cases*}
\alpha_\capital=\alpha_\labor=\frac{\eta}{g} & if $r\le \eta+2g\gamma$,\\
\alpha_\capital=\frac{\eta\gamma}{r-\eta-g\gamma}<\frac{\eta}{g}=\alpha_\labor& if $r>\eta+2g\gamma$.
\end{cases*}\label{eq:alphaexmp}
\end{equation}
Note that the labor income Pareto exponent $\alpha_\labor=\eta/g$ is highly sensitive to the income growth rate $g$. However, provided that $r>\eta+2g\gamma$, the capital income Pareto exponent $\alpha_\capital$ is not very sensitive to the value of $g$ because the denominator is $r-\eta-g\gamma$. This example is consistent with our result in Section \ref{subsec:results} that the capital Pareto exponent is smaller than the labor Pareto exponent but the two values are only weakly related.
\end{exmp}

\subsection{Numerical example}

We further examine the tail behavior of income and wealth using a numerical example of the income fluctuation problem \eqref{eq:IF}. Suppose that asset return is \iid lognormal, so $\log R\sim N((\mu-\sigma^2/2)\Delta,\sigma^2\Delta)$, where $\Delta>0$ is the length of one period, $\mu$ is the expected return, and $\sigma$ is volatility. Suppose every period the agent is ``promoted'' with some probability, so the income growth rate is
\begin{equation*}
G_{t+1}=Y_{t+1}/Y_t=\begin{cases*}
1 & with probability $1-p$,\\
\e^g & with probability $p$,
\end{cases*}
\end{equation*}
where $p\in (0,1)$ is the promotion probability and $g$ is the log income growth rate conditional on promotion. We parametrize the promotion probability as $p=1-\e^{-\Delta/L}$, where $L$ is the expected length of time until a promotion. Using \eqref{eq:BeareToda}, the labor income Pareto exponent is determined such that
\begin{equation}
1=v\E [G^{\alpha_Y}]=v(1-p+p\e^{g\alpha_Y})\iff \alpha_Y=\frac{1}{g}\log\frac{1-v+vp}{vp}.\label{eq:alphaY}
\end{equation}

We set the parameter values as in Table \ref{t:param}. One unit of time corresponds to a year and one period is a quarter, so $\Delta=1/4$. The preference parameters (discount rate and risk aversion) are standard. The death rate of $\eta=0.025$ implies an average (economically active) age of $1/\eta=40$ years. The expected return and volatility roughly correspond to the stock market. We set the labor income Pareto exponent to $\alpha_Y=3$, which is roughly the median value in Figure \ref{fig:alpha}. Using the survival probability $v=\e^{-\eta\Delta}$ and \eqref{eq:alphaY}, the implied value of income growth upon promotion is $g=0.0403$. The wealth Pareto exponent determined by \eqref{eq:BeareToda2} is then $\alpha=1.201$.

\begin{table}[!htb]
\centering
\caption{Parameter values}\label{t:param}
\begin{tabular}{lcc}
\toprule
Parameter & Symbol & Value \\
\midrule
Length of one period & $\Delta$ & 1/4\\
Discount rate & $\delta$ & 0.04\\
Relative risk aversion & $\gamma$ & 2\\
Death rate & $\eta$ & 0.025 \\
Expected return & $\mu$ & 0.07\\
Volatility & $\sigma$ & 0.15\\
Expected time to promotion & $L$ & 5\\
Labor income Pareto exponent & $\alpha_Y$ & 3\\
\bottomrule
\end{tabular}
\end{table}

To numerically solve the income fluctuation problem \eqref{eq:IF}, we discretize the log asset return $\log R$ using a 7-point Gauss-Hermite quadrature and apply policy function iteration (see Appendix \ref{sec:solution}). After solving the individual problem, we apply the Pareto extrapolation algorithm developed in \cite{Gouin-BonenfantTodaParetoExtrapolation} to accurately compute the stationary (normalized) wealth distribution. Finally, we also simulate an economy with $10^5$ agents. Figure \ref{fig:IFP} shows the results.

\begin{figure}[!htb]
\centering
\begin{subfigure}{0.48\linewidth}
\includegraphics[width=\linewidth]{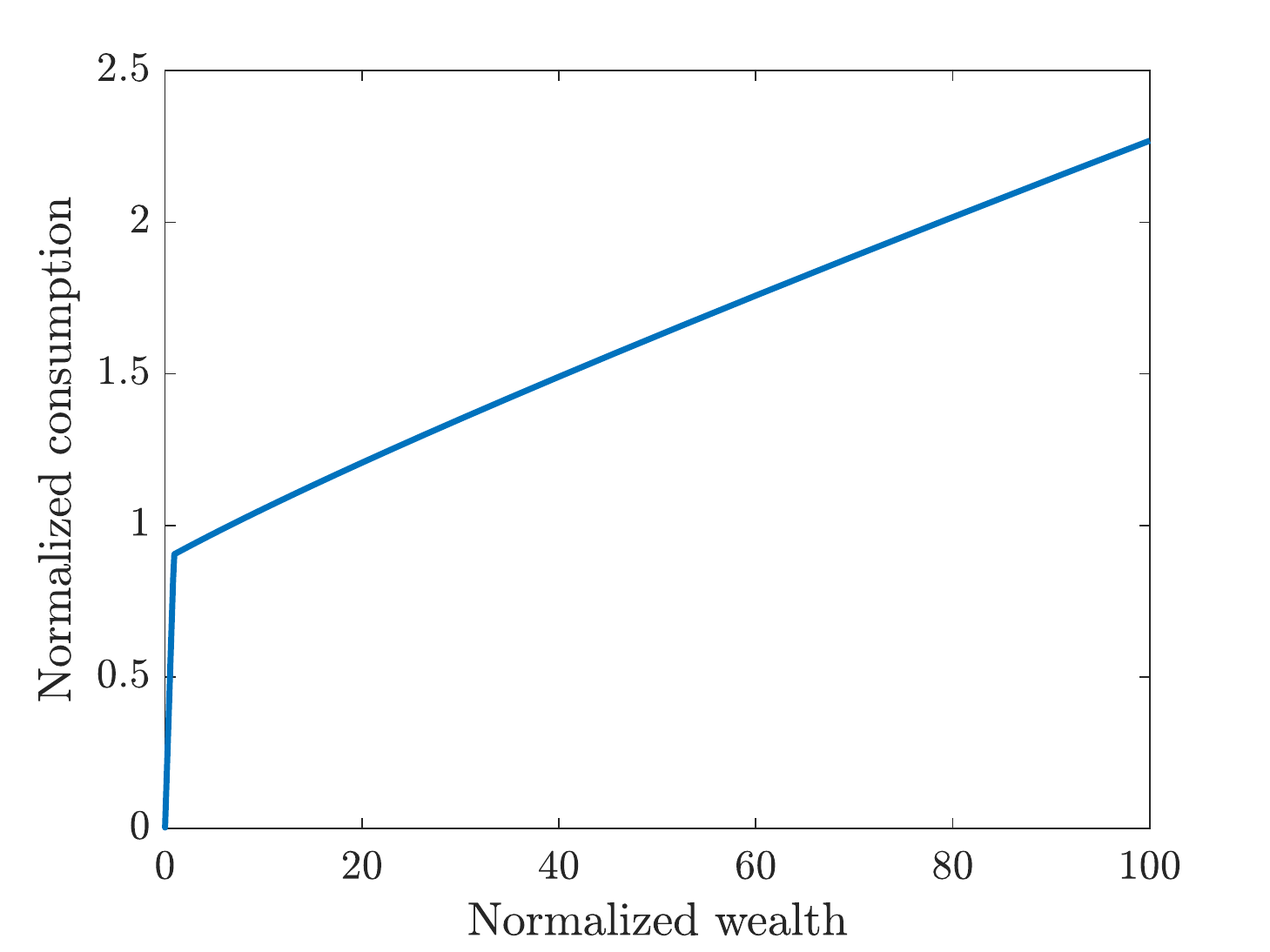}
\caption{Normalized consumption function.}\label{fig:IFP_ca}
\end{subfigure}
\begin{subfigure}{0.48\linewidth}
\includegraphics[width=\linewidth]{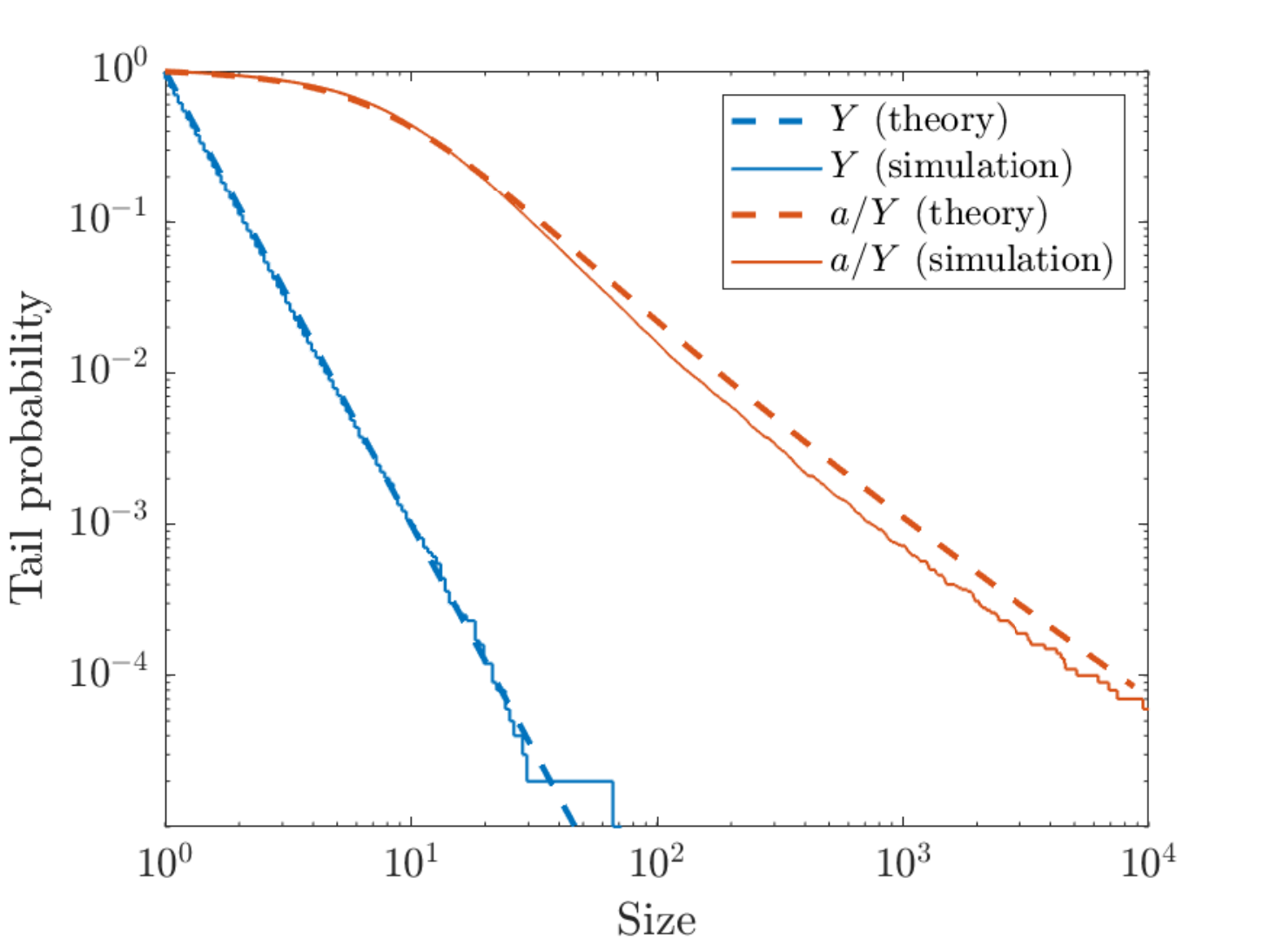}
\caption{Income and wealth distributions.}\label{fig:IFP_sim}
\end{subfigure}
\caption{Solution to income fluctuation problem.}\label{fig:IFP}
\end{figure}

Figure \ref{fig:IFP_ca} shows the normalized consumption function $\tilde{c}(\tilde{a})$ in the range $\tilde{a}\in [0,100]$. Consistent with Proposition \ref{prop:linear}, the consumption function is roughly linear for high asset level. Figure \ref{fig:IFP_sim} shows the size distributions of income $Y$ normalized wealth $\tilde{a}=a/Y$ in a log-log plot, both from the theoretical model and the simulation. The fact that the tail probability $\Pr(X>x)$ exhibits a straight line pattern in a log-log plot suggests that the size distributions have Pareto upper tails, consistent with theory. Furthermore, the slope for income is steeper than that of normalized wealth, so wealth (hence capital income) is more unequally distributed than labor income.

Finally, Figure \ref{fig:IFP_alpha} shows the income and wealth Pareto exponents when we change the income growth rate $g$ in the range $g\in [0.02,0.1]$, fixing other parameters. Because the income Pareto exponent is inversely proportional to income growth by \eqref{eq:alphaY}, the labor income Pareto exponent is highly sensitive to income growth. On the other hand, the wealth (capital income) Pareto exponent does not depend much on income growth by the same intuition as in Example \ref{exmp:IFP}. Thus our model is consistent with our empirical findings in Section \ref{subsec:results} that capital and labor income Pareto exponents are only weakly related.

\begin{figure}[!htb]
\centering
\includegraphics[width=0.7\linewidth]{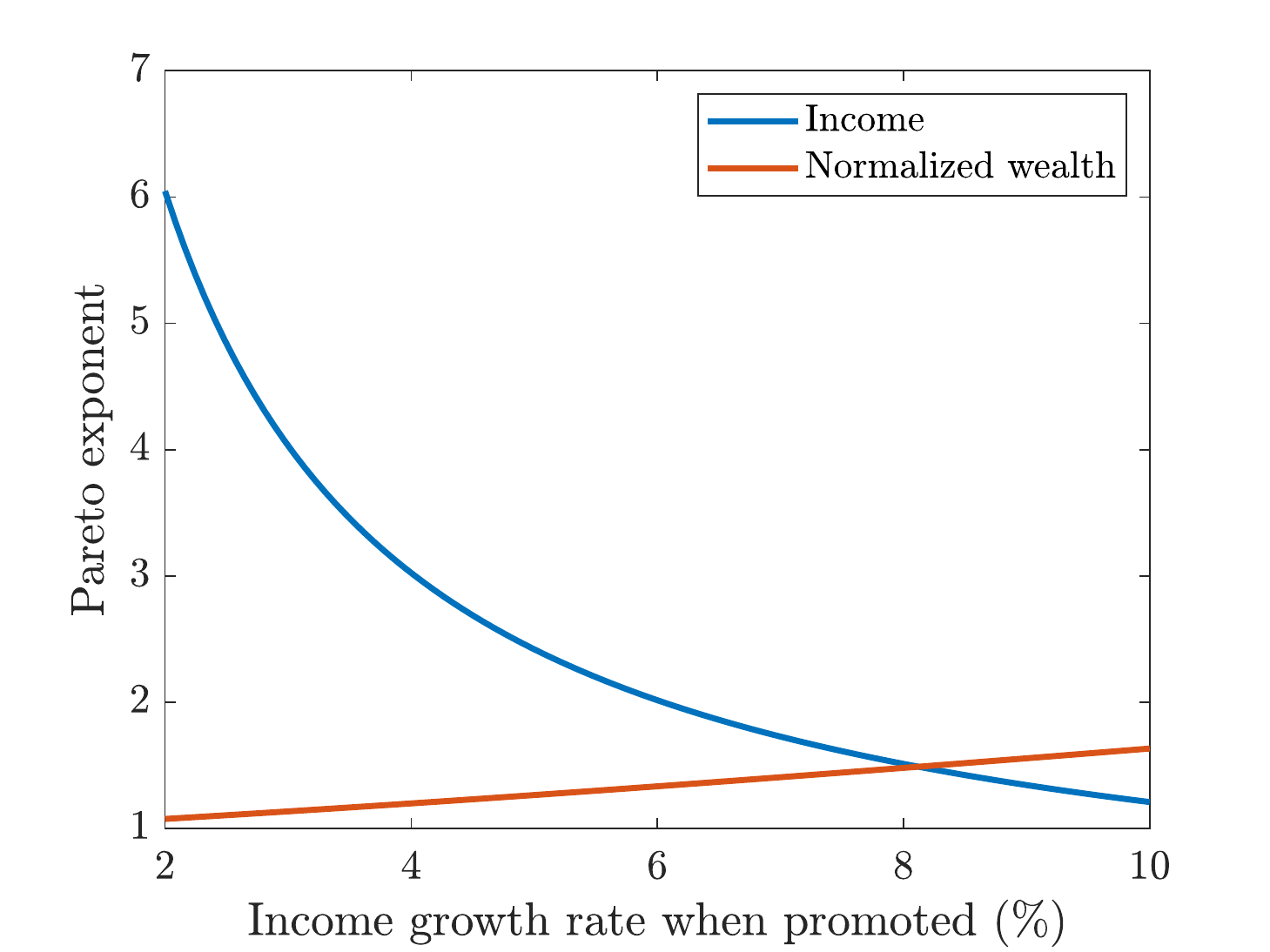}
\caption{Dependence of income and wealth Pareto exponents on $g$.}\label{fig:IFP_alpha}
\end{figure}



\appendix

\section{Proofs}\label{sec:proof}

\begin{proof}[Proof of Proposition \ref{prop:IF}]
Applying Theorem 2.2 of \cite{MaStachurskiToda2020JET} to the \iid case, a sufficient condition for the existence of a solution to the detrended problem is $\E [\tilde{\beta}]<1$ and $\E [\tilde{\beta}\tilde{R}]<1$, which is equivalent to \eqref{eq:spectral}.
\end{proof}

\begin{proof}[Proof of Proposition \ref{prop:linear}]
The concavity of $\tilde{c}$ follows from Proposition 2.5 and Remark 2.1 of \cite{MaStachurskiToda2020JET}. The asymptotic linearity of $\tilde{c}$ follows from \citet[Theorem 3]{MaToda2021JET}. Noting that
\begin{equation*}
\E [\tilde{\beta}\tilde{R}^{1-\gamma}]=\E [\beta G^{1-\gamma} (R/G)^{1-\gamma}]=\E [\beta R^{1-\gamma}],
\end{equation*}
the limit \eqref{eq:linear} follows from their Example 2.
\end{proof}

\begin{proof}[Proof of Proposition \ref{prop:Pareto_capital}]
Since by Proposition \ref{prop:linear} $\tilde{c}$ is concave, it is in particular Lipschitz continuous. Under the maintained assumptions, we can apply Theorem 1.8 of \cite{Mirek2011} to deduce that the normalized wealth $\tilde{a}$ is either bounded or has a Pareto upper tail with exponent characterized as the solution to \eqref{eq:BeareToda2}, where we have used the asymptotic linearity of $\tilde{c}$ established in Proposition \ref{prop:linear}.

By accounting, capital income (excluding capital loss) is
\begin{equation*}
Y_\capital\coloneqq \max\set{R-1,0}(a-c(a))=\max\set{R-1,0}Y(\tilde{a}-\tilde{c}(\tilde{a})).
\end{equation*}
Using Proposition \ref{prop:linear}, this quantity is approximately equal to $\rho \max\set{R-1,0}Y\tilde{a}$. The claim $\alpha=\min\set{\tilde{\alpha},\alpha_Y}$ then follows because asset return $R$ is thin-tailed and $a_t=Y_t\tilde{a}_t$ is the product of two (potentially dependent) random variables with Pareto upper tails, which inherits the smallest Pareto exponent by the result in \cite{JessenMikosch2006}.
\end{proof}

\section{Solving the income fluctuation problem}\label{sec:solution}
In this appendix we discuss how to solve the detrended income fluctuation problem. After detrending, the problem becomes
\begin{align*}
&\maximize && \E_0\sum_{t=0}^\infty \left(\prod_{s=1}^t \beta_t\right)\frac{c_t^{1-\gamma}}{1-\gamma}\\
&\st && a_{t+1}=R_{t+1}(a_t-c_t)+1,\\
&&& 0\le c_t\le a_t,
\end{align*}
where $\set{R_t,\beta_t}_{t=1}^\infty$ is \iid (though $R_t$ and $\beta_t$ are generally correlated.) According to \cite{MaStachurskiToda2020JET}, the Euler equation is
\begin{equation}
c_t^{-\gamma}=\max\set{\E_t [\beta_{t+1}R_{t+1}c_{t+1}^{-\gamma}],a_t^{-\gamma}}.\label{eq:Euler}
\end{equation}
Let $c(a)$ be the consumption function. Taking the $-1/\gamma$-th power of \eqref{eq:Euler}, we obtain
\begin{equation}
c(a)=\min\set{\left(\E [\beta R c(a')^{-\gamma}]\right)^{-1/\gamma},a},\label{eq:c_new}
\end{equation}
where $a'=R(a-c(a))+1$. Therefore we can compute the consumption function using the following variant of the policy function iteration algorithm:
\begin{enumerate}
\item Initialize the consumption function $c(a)$. For example, we can set $c(a)=\min\set{a,1+ma}$, where $m=\max\set{1-(\E [\beta R^{1-\gamma}])^{1/\gamma},0}$ is the theoretical asymptotic marginal propensity to consume according to \eqref{eq:linear}.
\item Update $c(a)$ by the right-hand side of \eqref{eq:c_new}, where $a'=R(a-c(a))+1$.
\item Iterate the above step until $c(a)$ converges.
\end{enumerate}
While the above algorithm has no guarantee to converge unlike the ``true'' policy function iteration algorithm discussed in \cite{MaStachurskiToda2020JET}, it has the advantage of avoiding root-finding and hence it is fast.

In Section \ref{sec:IFP}, we use this algorithm on a 100-point exponential grid for normalized wealth $\tilde{a}$ that spans $[0,10^4]$, with a median grid point of 10. The details on the exponential grid are discussed in \cite{Gouin-BonenfantTodaParetoExtrapolation}.

\section{Tables (not for publication)}\label{sec:tables}
\singlespacing

\begin{longtable}[c]{lcllrrllrr}
\caption{Point estimates of income Pareto exponents across countries and years. }\label{t:Pareto} \\
\toprule
& & \multicolumn{4}{c}{Labor income} & \multicolumn{4}{c}{Capital income} \\
Country & Year & $\hat{\alpha}$ & s.e. & $k$ & $N$ & $\hat{\alpha}$ & s.e. & $k$ & $N$ \\*
\cmidrule(lr){1-2}
\cmidrule(lr){3-6}
\cmidrule(lr){7-10}
\endhead
\bottomrule
\endfoot
\caption*{\footnotesize Note: $\hat{\alpha}$: point estimate of Pareto exponent; s.e.: standard error $(\hat{\alpha}(k)/\sqrt{k})$; $k$: the number of order statistics used for the Hill-estimator; $N$: the number of positive labor (resp.~capital) income observations. ``--'' indicates fewer than \Nmin positive observations so estimates are omitted.}

\endlastfoot
Australia & 1981 & 5.31 & 0.23 & 528 & 10568 & 2.1 & 0.1 & 474 & 9488 \\ 
Australia & 1985 & 4.28 & 0.26 & 275 & 5492 & 1.41 & 0.09 & 262 & 5232 \\ 
Australia & 1989 & 3.97 & 0.17 & 531 & 10629 & 1.47 & 0.06 & 549 & 10973 \\ 
Australia & 1995 & 3.1 & 0.2 & 229 & 4574 & 1.71 & 0.13 & 168 & 3364 \\ 
Australia & 2001 & 3.06 & 0.2 & 226 & 4510 & 1.47 & 0.12 & 155 & 3105 \\ 
Australia & 2003 & 3.45 & 0.19 & 337 & 6741 & 1.52 & 0.1 & 224 & 4473 \\ 
Australia & 2004 & 3.44 & 0.17 & 395 & 7893 & 1.42 & 0.08 & 307 & 6145 \\ 
Australia & 2008 & 3.07 & 0.17 & 333 & 6666 & 1.28 & 0.08 & 237 & 4749 \\
\bottomrule
\end{longtable}

\begin{longtable}[c]{@{}lcllllrrc@{}}
\caption{Hypothesis testing of $\hat{\alpha}_\labor = \hat{\alpha}_\capital$.}\label{t:equal} \\
\toprule
Country & Year & $\hat{\alpha}_\labor$ & s.e. & $\hat{\alpha}_\capital$ & s.e. & $k$ & $N$ & Reject $H_0$?\\
\midrule
\endhead
\bottomrule
\endfoot

\caption*{\footnotesize Note: $\hat{\alpha}_\labor$ (resp.~$\hat{\alpha}_\capital$): point estimate of Pareto exponent for labor (resp.~capital) income; s.e.: standard error $(\hat{\alpha}(k)/\sqrt{k})$; $k$: the number of order statistics used for the Hill-estimator; $N$: the number of positive capital income observations; the last column denotes if $H_0$ is rejected.}
\endlastfoot

Australia & 1981 & 5.26 & 0.24 & 2.1 & 0.1 & 474 & 9488 & Yes \\ 
Australia & 1985 & 4.22 & 0.26 & 1.4 & 0.09 & 262 & 5232 & Yes \\ 
Australia & 1989 & 3.95 & 0.17 & 1.46 & 0.06 & 531 & 10629 & Yes \\ 
Australia & 1995 & 2.91 & 0.22 & 1.7 & 0.13 & 168 & 3364 & Yes \\ 
Australia & 2001 & 3.24 & 0.26 & 1.45 & 0.12 & 155 & 3105 & Yes \\ 
Australia & 2003 & 3.25 & 0.22 & 1.51 & 0.1 & 224 & 4473 & Yes \\ 
Australia & 2004 & 3.38 & 0.19 & 1.41 & 0.08 & 307 & 6145 & Yes \\ 
Australia & 2008 & 3.3 & 0.21 & 1.27 & 0.08 & 237 & 4749 & Yes \\
\bottomrule
\end{longtable}

\end{document}